\documentclass[titlepage,12pt,a4paper]{article}

        \usepackage{mathtools,amsthm,amssymb}

        \usepackage{natbib}

        \usepackage{hyperref}

        \hypersetup{colorlinks=true,urlcolor=blue,citecolor=red,linkcolor=red}

        \usepackage[center]{caption}
	
	\usepackage{setspace}
 	
	\usepackage{epsfig}
	
	\usepackage{graphics}
	
	\usepackage{lscape}

	\usepackage[english]{babel}
	
	\usepackage{color}
	
        \theoremstyle{plain}\newtheorem {claim}{Claim}

        \theoremstyle{plain}\newtheorem {corollary}{Corollary}

        \theoremstyle{plain}

        \theoremstyle{plain}\newtheorem {proposition}{Proposition}

      	\theoremstyle{plain}\newtheorem {theorem}{Theorem}
      	\theoremstyle{plain}\newtheorem *{theorem*}{Theorem}

        \theoremstyle{definition}\newtheorem {assumption}{Assumption}

        \theoremstyle{definition}

        \theoremstyle{remark} \newtheorem{example}{Example}

        \theoremstyle{remark} 

        \theoremstyle{remark}

\begin{document}

        \bibliographystyle{plainnat}

        \bibpunct{(}{)}{;}{a}{,}{,}

        \title{Affective interdependence and welfare} 

        \author{A. Heifetz \thanks{Open Un. Israel, \href{mailto:aviadhe@openu.ac.il} {\tt aviadhe@openu.ac.il}} \and E. Minelli \thanks{Un. Brescia, \href{mailto: enrico.minelli@unibs.it} {\tt enrico.minelli@unibs.it}} \and H. Polemarchakis \thanks{Un. Warwick,  \href{mailto:h.polemarchakis@warwick.ac.uk} {\tt h.polemarchakis@warwick.ac.uk}}}

        \date{April 30, 2023}

       \maketitle

        \begin{abstract}

  	\textit{Purely affective interaction} allows the welfare of an individual to depend on her own actions and on the profile of welfare levels of others. Under an assumption on the structure of mutual affection that we interpret as ``non-explosive mutual affection,'' we show that equilibria of simultaneous-move affective interaction are Pareto optimal independently of whether or not an induced standard game exists. Moreover, \emph{if} purely affective interaction induces a standard game, then an equilibrium profile of actions is a Nash equilibrium of the game, and this Nash equilibrium and Pareto optimal profile of strategies is locally dominant.

        	\bigskip 
        
       	\bigskip


        \bigskip


        	\end{abstract}
   
   	\maketitle

\section{Introduction}

\textit{Purely affective interaction} allows the welfare of an individual to depend on her own actions and on the profile of welfare levels of others. Importantly, actions of others do not affect directly the welfare of an individual. Here, we provide a concise and general treatment of the class of smooth, purely affective interactions.  Affection can be positive or negative. We focus on welfare implications.

\citet*{rayvohra20} demonstrated a striking result:  if a purely affective interaction induces a standard game, Nash equilibria of the induced game are Pareto optimal. Which proved a conundrum, since earlier results, notably by \citet*{arrow81} in gift-giving interaction with three individuals in a simultaneous-move setting and \citet*{pearce08} in a cake-eating game in a sequential-move setting, had concluded that affective interaction is not conducive to optimality. We resolve this conundrum by showing that optimality emerges in settings when interaction is purely affective. When  inefficiency arises, dependence of individual welfare on actions of other individuals lurks in the background either via the game form, as in Arrow's gift giving game\footnote{\citet*{bourlesetal17}} or via aggregate feasibility constraints as in Pearce's cake eating problems or, more generally, competitive economies. Indeed, in a competitive equilibrium setting, \citet*{winter69} and \citet*{sobeletal11} allowed individuals to have preferences over their consumption, and, concurrently, over the profile of utilities. They identified a condition, \textit{social monotonicity} that, under assumptions of monotonicity and convexity in own consumption and, in particular, separability of utilities between own utility and the profile of utilities of others, implies that Pareto optimal allocations can be supported as competitive allocations and can be attained with redistributions of revenue. They also showed by example that social monotonicity \emph{does not} guarantee the efficiency of competitive equilibria. 

An assumption that we maintain throughout is that a transformation of the matrix of mutual affection is a \textit{P-matrix,} as in \citet*{galenikaido65}, which we interpret as non-explosive mutual affection.
Under this assumption, we show that equilibria of simultaneous-move \emph{purely} affective interaction are Pareto optimal independently of whether or not an induced standard game exists. 
Moreover, \emph{if} purely affective interaction induces a standard game, then an equilibrium profile of actions is a Nash equilibrium of the game, and  this Nash equilibrium and Pareto optimal profile of strategies is locally dominant\footnote{Optimality results indicate that games induced by purely affective interaction form a non-generic class within the class of games: Nash equilibria of generic games are suboptimal in \citet*{dubey86}, and, likewise, backward-induction paths of generic sequential games are suboptimal in \citet*{heifetzetal21}.}. For the sequential setting, \citet*{heifetz22} defined backward induction much more simply and directly than in \citet*{pearce08}, and he showed that backward induction paths of actions and utility levels are Pareto optimal, again under the assumptions of non-explosive \emph{purely} affective interaction.

\section{Purely affective interaction}

Individuals are $ i\in I=\left\{  1,...,n\right\}, $ profiles of action are
\[
x=\left(  x\,_{1},...,x_{n}\right)  \in X=\prod_{i=1}^{n}X_{i},
\]
and profiles of utility levels are
\[
u=\left(  u_{1},...,u_{n}\right)  \in\mathbb{R}^{n}\footnote{We employ the standard notation $u\leq\bar{u}$ for $u_{i}\leq\bar{u}_{i}$,
$i=1,..,n$ \ and $u<\bar{u}$ for $u_{i}\leq\bar{u}_{i}$, $i=1,..,n$ with at
least one strict inequality.}.
\] 
Utility functions are
\[
V\left(  x,u\right)  =\left(  V_{i}\left(  x_{i},u_{-i}\right):{i=1,...n} \right),
\]
and we also write
\[
V_{x}\left(  u\right)  =V\left(  x,u\right).
\]

\bigskip

\bigskip

A profile of actions and utility levels  $ (x, u) $ is \textit{consistent} if, for every individual, $ u_{i} $ corresponds to the utility level at $(x_i,  u_{-i}) $ or
\[
V_{x}\left( u \right)  = u.
\]

In a purely affective interaction, an individual controls her action $x_i$,  and is aware of the direct effect of her choice on her well-being. Other individuals' choices do not affect her directly\footnote{Maybe because every externality has been priced -- see the last section of the paper. }, but she cares about the well being of others, and her perception of others' well-being may affect her preferences over her actions. 
The understanding of others' action sets and preferences is too imprecise to allow an individual to calculate the possible \emph{indirect} effect of her choices of their well-being. At a consistent profile of  actions choices and utility levels, perceptions are confirmed.

A  \textit{parametric equilibrium} is a consistent profile of actions and utility levels,  $(x^*, u^*) $ that satisfies individual optimization: every individual maximizes $V_i$ taking $u_{-i}^*$ as given or
\[
V_{x}\left(  u^{\ast}\right)  \leq V_{x^{\ast}}\left(  u^{\ast}\right).
\]


A  profile of actions and utility levels, $\left(  \tilde{x},\tilde{u}\right)  $ is a \textit{Pareto improvement}\ over a profile $\left(  x, u \right)  $ if 
\[
\tilde{u}>u.
\]

A consistent profile of actions and utility levels, $ (x, u) $ is \textit{Pareto optimal} if  it does not permit a consistent Pareto improvement.

\begin{assumption}
For every individual, $X_{i}$ is an open subset of Euclidean space, and the utility function $V_{i}(\cdot, \cdot)  $ is twice continuously differentiable.
\end{assumption}

The Jacobian of $V_{x}$ at $u$ is $J_{x}\left(  u\right).  $ 

\bigskip

A square matrix is  a \textit{P-matrix} if all  its principal minors are positive\footnote{A principal minor is obtained by the elimination of rows and corresponding columns, but, importantly, without transpositions of rows or columns prior to elimination.}. 

\begin{assumption}
For every $x\in X$ and $u\in\mathbb{R}^{n}$, the matrix $(I -  J_{x}\left(  u\right) )$ is a P-matrix.
\end{assumption}

To interpret the assumption, we recall a useful characterization that we use repeatedly.

\bigskip
\noindent
\textbf{Gale-Nikaido Lemma} [\citet*{galenikaido65}, Theorem 2]
\textit{A matrix A is a P-matrix if and only if, for any non-zero $y \in \mathbb{R}^n$, there exists $i \in \{1,2, \ldots n\},$ such that $y_i (Ay)_i > 0$.}

\bigskip

In words, P-matrices do not fully reverse the sign of any non-zero vector. 

\bigskip

In our context, this property allows us to interpret Assumption 2 as an assumption of \textit{non-explosive mutual affection.} To see this, notice that the Gale-Nikaido characterization allows us to rewrite Assumption 2.

For every $x$ and $u$ and all $\Delta u \neq 0,$  there exists an $i,$ such that
\[\Delta u_i > 0 \quad  \text{and} \quad v  \Delta u_i > \sum_{j \neq i} \frac{\partial V_{x, i}(x_i, u)}{\partial u_j}\Delta u_j\]
or
\[\Delta u_i < 0  \quad  \text{and} \quad  \Delta u_i < \sum_{j \neq i} \frac{\partial V_{x, i}(x_i, u)}{\partial u_j}\Delta u_j.\]

Suppose now that we start from a consistent pair $(\hat x, \hat u)$ at which
$$\hat u - V_{\hat x} (\hat u) =0. $$
Then, for any exogenous change in utility levels while holding the profile of actions fixed,  $\Delta u = u - \hat u \neq 0$,  there is one individual $i,$ for whom 
\[ u_i > \hat u_i, \quad \text{and} \quad V_i(\hat x, u) < u_i \]
or 
\[u_i <  \hat u_i \quad \text{and} \quad V_i(\hat x, u) > u_i.\]
That is, under Assumption 2, starting from a consistent pair of actions and utility levels, for any exogenous change in the utility levels there is always one individual whose resulting utility, after the change, does  not reinforce the exogenous change. Thus, Assumption 2 allows for a wide array of positive and negative individual feelings about changes in the well-being of (any subset) of other individuals, but it prevents explosive affective interaction. We shall come back to the interpretation of Assumption 2 later, when discussing alternative, stronger restrictions. Here we  record an important  implication.

\begin{theorem}
Under Assumptions $1$ and $2$, if $\left(  x^{\ast},u^{\ast}\right)  $ is a parametric equilibrium, then the consistent allocation of utility levels, $ u^{\ast} $ is Pareto optimal.
\end{theorem}

\begin{proof} Suppose, by way of contradiction, that $\left(
\tilde{x},\tilde{u}\right)  $ Pareto improves on $\left(  x^{\ast},u^{\ast}\right), $
\begin{equation}\label{impr}
V_{\tilde{x}}\left(  \tilde{u}\right)  =\tilde{u}>u^{\ast}. 
\end{equation}
Let $F:\mathbb{R}^{n}\rightarrow\mathbb{R}^{n}$ be defined by
\[
F\left(  u\right)  =u-V_{\tilde{x}}\left(  u\right).
\]
Since $\left(  x^{\ast},u^{\ast}\right)  $ is a parametric equilibrium, 
\[
F\left(  u^{\ast}\right)  =u^{\ast}-V_{\tilde{x}}\left(  u^{\ast}\right)
\geq0,
\]
and 
\begin{equation}\label{pareq}
F\left(  \tilde{u}\right)  =\tilde{u}-V_{\tilde{x}}\left(  \tilde{u}\right)
=0\leq F\left(  u^{\ast}\right). 
\end{equation}
By Assumptions 1 and 2,  the matrix $ (I-J_{\tilde{x}}\left(  u\right)),  $ the Jacobian of $F$, 
is a P-matrix and,  by Theorem 3 of \citet*{galenikaido65},
the inequalities $\left(  \ref{impr} \right)  $ and $\left(  \ref{pareq}\right)  $ cannot obtain
simultaneously for $\tilde{u}\neq u^{\ast}$. 
\end{proof}

\subsection*{The case of linearly separable affection}

In a linearly separable purely affective interaction, the individuals' utility functions
have the form
\begin{equation}\label{lin}
V_i(x_i, u) = f_i(x_i) + \sum_{j \neq i} a_{ij} u_j. 
\end{equation}

At every $x=\left(  x_{1},...,x_{n}\right)  $ the Jacobian $J$ of $V$ with
respect to $u=\left(  u_{1},...,u_{n}\right)  $ has a zero diagonal $J_{ii}=0$
and off-diagonal entries $J_{ij}=a_{ij}$. 

Consistency in this special case  takes the form
\begin{equation}\label{lincon}
u = f(x) + Ju.
\end{equation}
Under Assumption 2,  $\det\left(  I-J\right)  \neq 0, $ and we can uniquely solve the system of equations (\ref{lincon}) at every $x$, thus  obtaining the \textit{induced game} corresponding to (\ref{lin}),
$$
U(x) = (I - J)^ {- 1} f(x)= Bf(x).
$$
The utility functions
$U=\left(  U_{1},...,U_{n}\right)  $ in the  induced game are  linear
combinations of the `base utilities' $f=\left(  f_{1},...f_{n}\right)  $.  The matrix $B= (I - J)^ {- 1} $ summarizes the effect of changes in the base utilities $f(x)$ on the final well being of individuals, taking into account the network of affective interactions between them.

Under Assumption 2,  $B = (I -J)^{-1}$ is also a P-matrix\footnote{\citet*{hornjohnson13}, Theorem 4.3.2}.
In particular,  the diagonal of $B$ is positive, and, for every i, 
\begin{equation}\label{P}
\\ b_{ii} = \frac{\mid (I -J)_{ii} \mid}{Det(I - J)} > 0.
\end{equation}

\begin{claim}
Under Assumption 2, in a situation of linearly separable affection, an action profile $x^*$ is a Nash equilibrium of the induced game $U$  if and only if  $(x^*, U(x^*))$ is a  parametric equilibrium of $V$.
\end{claim}

A planner may try to obtain a Pareto improvement over a Nash equilibrium in the induced game $U$ by \emph{jointly} changing the actions  of everybody. In the case of linearly separable affection, a Pareto improvement is a $\Delta f_x $ such that $\Delta U = B \Delta f_x > 0$, while a Nash deviation in the induced game is a $\Delta f^i_x = ( 0, \ldots, df^i, \ldots, 0)$ such that $\Delta U^i = B  \Delta f^i_x> 0$.
If $x^*$ is  a Nash equilibrium of $U$ and $y=\Delta f_{x^*} $  a Pareto improvement, then $By  > 0$ and $ y \neq 0.$ But then,  by  the Gale-Nikaido Lemma, there must exist $i$ such that $y^i = \Delta f^i_{x^*} > 0$ (also, given (\ref{P}), $ B\Delta f^i_{x^*} >0$), a contradiction with $x^*$ being a Parametric equilibrium (also, with $x^*$ being Nash equilibrium). 

\begin{claim}
Under Assumption 2, in a situation of linearly separable affection,  if $x^*$ is a Nash equilibrium of the induced game $U$, then $x^*$  is Pareto optimal.
\end{claim}

\bigskip

For given $\lambda \in \mathbb{R}_{+}^{n}\setminus \{0\}$, consider the welfare function
$$
W_\lambda (x)= \lambda U(x) =  \lambda Bf(x), 
$$
a linear combination of the 'base utilities' $f_i(x_i)$. The maximization of each  $f_i(x_i)$ thus assures that the first order conditions for the maximization of $W_\lambda (x)$ are satisfied. 
Let us assume concavity of base utilities.
\begin{assumption}\label{coa}
For every $i$, the function $f_i$ is concave in $x_i$.
\end{assumption}
Even under Assumption 3,  $W_\lambda $  need not be concave in $x$, because Assumption 2 does not guarantee that the elements of $B$ are non-negative. 

Using Farkas's Lemma, \citet*{galenikaido65} (Corollary 2) prove  that
\[
B \;  \text{is a P-matrix}  \;\; \implies \exists \; \lambda \in \mathbb{R}_{++}^{n} \; \; s.t. \; \;  
   \lambda B  > >  0.
\]
Therefore, under Assumptions 2 and 3,  \emph{there exist} welfare weights $\lambda$ such that \emph{for those weights} $W_\lambda$ is  a sum of concave functions, and therefore concave, leading to an alternative proof of Pareto optimality of Nash equilibrium.

\begin{claim}
Under Assumptions 1, 2 and 3, in a situation of linearly separable affection,  if $x^*$ is a Nash equilibrium of the induced game $U$, there exists $\lambda \in \mathbb{R}_{++}^{n}$ such that  $x^*$ is  a global maximum of $W_\lambda$.
\end{claim}

\bigskip

\begin{example}[Two person linearly separable affection]

Consider the two person purely affective interaction where
\begin{align*}
V_{1}\left(  x,u_{2}\right)   &  =f\left(  x\right)  +au_{2},\\
V_{2}\left(  y,u_{1}\right)   &  =g\left(  y\right)  +bu_{1},
\end{align*}
with
\[
ab<1.
\]
That is, the two individuals can have positive or negative feelings towards each other, and these feelings may even be strong, but they satisfy the assumption of moderate \emph{reciprocal} affection: if feelings go in the same direction (both love or both spite) they cannot be too strong.

\bigskip

For each $\left(  x,y\right)  $ the Jacobian of the map $V_{\left(
x,y\right)  }:\mathbb{R}^{2}\rightarrow\mathbb{R}^{2}$ is%
\[
J_{\left(  x,y\right)  }=\left(
\begin{array}
[c]{cc}%
0 & a\\
b & 0
\end{array}
\right),
\]
and the matrix $(I - J_{\left(  x,y\right)  })$ has a unitary  diagonal and determinant $\det (I - J_{\left(  x,y\right)  }) =1 - ab > 0$,   and therefore it is a P-matrix\footnote{Such separable interactions in which $\left\vert a\right\vert >1$ (and $\left\vert b\right\vert <\frac
{1}{\left\vert a\right\vert }$) \ do not induce \textquotedblleft a game of
love and hate\textquotedblright\ in the sense of \citet*{rayvohra20}, and
therefore are not covered by their analysis, because their boundedness
condition (i) is not satisfied -- for every $\left(  x,y\right)  $ and every
function $B\left(  x,y\right)  <\infty,$ whenever $\left\vert u_{2}\right\vert
>B\left(  x,y\right)  +\left\vert f\left(  x\right)  \right\vert $, in the
$\sup$ norm $\left\Vert \cdot\right\Vert, $
\[
\left\Vert U\left(  \left(  x,y\right)  ,\left(  u_{1},u_{2}\right)  \right)
\right\Vert \geq\left\vert f\left(  x\right)  +au_{2}\right\vert
\geq\left\vert u_{2}\right\vert -\left\vert f\left(  x\right)  \right\vert
>B\left(  x,y\right).
\]
}.

\bigskip

The induced game is 
\begin{align*}
U_{1}\left(  x,y\right)   &  =\frac{f\left(  x\right)  +ag\left(  y\right)
}{1-ab},\\
U_{2}\left(  x,y\right)   &  =\frac{g\left(  y\right)  +bf\left(  x\right)
}{1-ab},
\end{align*}
whose Nash equilibria (if there are any) are $\left(  x^{\ast},y^{\ast
}\right)  $ where
\[
x^{\ast}\in\arg\max f\left(  x\right)  ,\quad y^{\ast}\in\arg\max g\left(
y\right).
\]
That is, every Nash equilibrium is in dominant strategies.

In particular, if $f$ and $g$ are concave then there exists at most one Nash
equilibrium $\left(  x^{\ast},y^{\ast}\right)$. A social welfare function of
the form
\[
W\left(  x,y\right)  =\lambda_{1}U_{1}\left(  x,y\right)  +\lambda_{2}%
U_{2}\left(  x,y\right),
\]
where $\left(  \lambda_{1},\lambda_{2}\right)  > 0,$ is then concave and
maximized at the unique Nash equilibrium $\left(  x^{\ast},y^{\ast}\right)$  if
and only if 
\[
\frac{1}{1-ab}\left(
\begin{array}
[c]{cc}%
\lambda_{1} & \lambda_{2}%
\end{array}
\right)  \left(
\begin{array}
[c]{cc}%
1 & a\\
b & 1
\end{array}
\right)  >0.
\]
With $ab  <1$ \ and therefore $\frac{1}{1-ab}>0,$ such
$\left(  \lambda_{1},\lambda_{2}\right)  >0$ indeed exists since

(i) if both $a\geq0$ and $b\geq0$ (mutual sympathy) then $U_{1}$ and $U_{2}$
are already concave themselves, the Nash equilibrium $\left(  x^{\ast}%
,y^{\ast}\right)  $ is their global maximum, and any $\left(  \lambda
_{1},\lambda_{2}\right)  >0$ will do,

(ii) if $a\geq0$ but $b<0$ (individual 1 is sympathetic and individual 2 is spiteful)
then $U_{1}$ is concave and the Nash equilibrium $\left(  x^{\ast},y^{\ast
}\right)  $ is its global maximum, while $U_{2}$ is not concave, and $\left(
x^{\ast},y^{\ast}\right)  $ is a saddle point of it. One can then choose
$\lambda_{1}>0$ \ and $\lambda_{2}=0$ to get a concave welfare function $W$
which is maximized at $\left(  x^{\ast},y^{\ast}\right),  $

(iii) similarly, if $b\geq0$ but $a<0$ \ one can choose $\lambda_{2}>0$ \ and
$\lambda_{1}=0$ to the desired effect, and

(iv) finally, if both $a<0$ and $b<0$ (mutual spite) then both $U_{1}$ and
$U_{2}$ are not concave, and the Nash equilibrium $\left(  x^{\ast},y^{\ast
}\right)  $ is a saddle point of each of them. \ Nevertheless, since $ab<1$
(the reciprocal extent of spite is moderate), the vectors
$\left(  1,b\right)  $ and $\left(  a,1\right)  $ are within a half-plane
containing the positive orthant, and therefore both $\left(  1,b\right)  $ and
$\left(  a,1\right)  $ form an acute angle with vectors $\left(  \lambda
_{1},\lambda_{2}\right)  >0$ in some positive cone\footnote{This cone becomes
narrower as $ab\nearrow1$. \ The weights $\left(  \lambda_{1},\lambda
_{2}\right)  >0$ in this cone `strike the right balance' between the
curvatures of $U_{1}$ and $U_{2}$ at $\left(  x^{\ast},y^{\ast}\right)  $ --
between the concavity of $U_{1}$ in $x$ and the convexity of $U_{2}$ in $x$ so
that the linear combination $W$ -- the social welfare function -- is concave
in $x$, and \textit{at the same time} also between the convexity of $U_{1}$ in
$y$ and the concavity of $U_{2}$ in $y$, so that the linear combination $W$ is
concave also in $y$.}.

We therefore conclude that \textit{in all cases}, the Nash equilibrium
$\left(  x^{\ast},y^{\ast}\right)  $ is Pareto optimal.  

\end{example}

\subsection*{Locally induced game}

In the general, non additively separable case, Assumption 2 does not guarantee that at any given $x$ the system of equations 
\begin{equation}\label{ldg}
F_x(u) = u - V_x(u) = 0
\end{equation}
has a solution, so the induced game $U(x)$ need not be well defined everywhere on $X$.

Still,  under Assumption 2, another result of \citet*{galenikaido65}, Theorem 4, implies that \emph{if} a solution $u_x$ of (\ref{ldg}) exists, then it is unique.

Also, again by Assumption 2, $\det (I - J_x(u_x))  \neq 0$, and we can apply  the implicit function theorem to $F: X \times \mathbb{R}^n \to  \mathbb{R}^n$ at $\left(  x, u_x\right) $
to obtain the existence of smooth real-valued utility
functions $U_x\left(  \cdot\right)  =\left(  U_{i}\left(  \cdot\right)  \right)
_{i=1,...,n}$  defined on some neighborhood $\mathcal{O}_x$ of $x$ with
\[
U_x\left(x\right)  =u_x,
\]
and
\begin{equation}\label{focldg}
\frac{\partial U_{i}\left(  x \right)  }{\partial x_{j}}=\frac{\partial
V_{j}\left( x_{i},u_{x -i}\right)  }{\partial x_{j}}\left(  \left(
I-J_{x}\left(  u_x \right)  \right)  ^{-1}\right)  _{ij}.
\end{equation}
We call $U_x : \mathcal{O}_x \to \mathbb{R}^n$ the \textit{locally induced game} of $V$ at $x$.

We now derive analogs of Claims 1 and 2 for the general case of purely affective interactions.

\begin{theorem}
Under Assumptions 1 and 2, if $\left(  x^{\ast},u^{\ast}\right)  $ is a parametric equilibrium of $V$, 
then $x^{\ast}$ is a Nash equilibrium of the locally induced game  $U_{x^*}$.
\end{theorem}

\begin{proof} Suppose, by way of contradiction, that $x^{\ast}$ is
not a Nash equilibrium of the locally induced game. Then, for some individual $i,$ there
exists an alternative choice $\tilde{x}_{i}\in X_{i}$ for which the locally induced
game is defined at
\[
x=\left(  \tilde{x}_{i},x_{-i}^{\ast}\right),
\]
and
\[
\tilde{u}_{i}=U_{i}\left(  x\right)  >U_{i}\left(  x^{\ast}\right)
=u_{i}^{\ast}
\]
Where, to simplify notation, use $U$ for the locally induced game $U_{x^*}$.  

\ Let
\[
\tilde{u}=U\left(  x\right)  =\left(  V_{j}\left(  x_{j},\tilde{u}%
_{-j}\right)  \right)  _{j=1}^{n},
\]
\ and let $F:\mathbb{R}^{n}\rightarrow\mathbb{R}^{n}$ be defined by
\[
F\left(  u\right)  =u-V_{x}\left(  u\right).
\]
By the definition of $\tilde{u},$ 
\[
F\left(  \tilde{u}\right)  =0.
\]
Also, for $j\neq i$, since $x_{j}=x_{j}^{\ast},$ 
\[
F_{j}\left(  u^{\ast}\right)  =0.
\]
At the same time, since $\left(  x^{\ast},u^{\ast}\right)  $ is a parametric
equilibrium and $x_{i}\neq x_{i}^{\ast},$
\[
F_{i}\left(  u^{\ast}\right)  \geq0.
\]
Altogether,
\[
\left(  u_{k}^{\ast}-\tilde{u}_{k}\right)  \left(  F_{k}\left(  u^{\ast
}\right)  -F_{k}\left(  \tilde{u}\right)  \right)  \leq0,\qquad k=1,...,n.
\]
However, since by Assumption 2 the Jacobian $I-J_{x}\left(  u\right)  $ of
$F\left(  u\right)  $ is a P-matrix for every $u\in\mathbb{R}^{n}$, by theorem
20.5 in \citet*{nikaido68}, this set of inequalities cannot obtain for $\tilde
{u}\neq u^{\ast}.$

\end{proof}

\bigskip

\bigskip

\emph{If} the induced game is well defined everywhere, the converse of Theorem 2 holds as well:

\begin{theorem} Suppose the induced game $U$ is defined on the entirety of $X$.
Under Assumptions 1 and 2, if $x^{\ast}$ is a Nash equilibrium of the induced game, then $\left(
x^{\ast},U\left(  x^{\ast}\right)  \right)  $ is a parametric equilibrium.
\end{theorem}

\begin{proof} Denote $u^{\ast}=U\left(  x^{\ast}\right)  $ and
suppose, by way of contradiction, that \newline  $\left(  x^{\ast},u^{\ast}\right)  $ is
not a parametric equilibrium. Then, for some individual $i,$ there exists an
alternative choice $\hat{x}_{i}\in X$ with
\[
\hat{u}_{i}=V_{i}\left(  \hat{x}_{i},u_{-i}^{\ast}\right)  >
V_{i}\left(x_{i}^{\ast},u_{-i}^{\ast}\right)  =u_{i}^{\ast}.
\]
\ By assumption, the induced game is defined at
\[
x=\left(  \hat{x}_{i},x_{-i}^{\ast}\right).
\]
Let \
\[
\hat{u}=U\left(  x\right)  =\left(  V_{j}\left(  x_{j},\hat{u}_{-j}\right)
\right)  _{j=1}^{n}
\]
and let $F:\mathbb{R}^{n}\rightarrow\mathbb{R}^{n}$ be defined by
\[
F\left(  u\right)  =u-V_{x}\left(  u\right).
\]
By the definition of $\hat{u},$
\[
F\left(  \hat{u}\right)  =0.
\]
For $j\neq i$, since $x_{j}=x_{j}^{\ast},$
\[
F_{j}\left(  u^{\ast}\right)  =0.
\]
At the same time, since $\left(  x^{\ast},u^{\ast}\right)  $ is a parametric
equilibrium, and $x_{i}\neq x_{i}^{\ast}$,
\[
F_{i}\left(  u^{\ast}\right)  \geq0.
\]
Altogether,
\[
\left(  u_{k}^{\ast}-\hat{u}_{k}\right)  \left(  F_{k}\left(  u^{\ast}\right)
-F_{k}\left(  \hat{u}\right)  \right)  \leq0,\qquad k=1,...,n.
\]
However, since by Assumption 2 the Jacobian $I-J_{x}\left(  u\right)  $ of
$F\left(  u\right)  $ is a P-matrix for every $u\in\mathbb{R}^{n}$, by theorem
20.5 in Nikaido (1968) this set of inequalities cannot obtain for $\hat{u}\neq
u^{\ast}.$
\end{proof}

Theorems 1 and 3 imply the analog of Claim 2:

\begin{corollary}
Suppose the induced game $U$ is defined on the entirety of $X$. Under Assumptions 1 and 2, if $x^*$ is a Nash equilibrium of the induced game $U$, then it  is Pareto optimal.
\end{corollary}

As a last remark on the special structure of purely affective interaction, we prove that, while in the linearly separable case  parametric equilibrium strategies are by construction dominant strategies, even in the general case parametric equilibrium strategies are \emph{locally} dominant.

\begin{theorem}
At a parametric equilibrium $\left(  x^{\ast},u^{\ast}\right),  $
each individual's action is locally dominant in the locally induced game $U_{x^*}$.
\end{theorem}

\begin{proof} At a parametric 
equilibrium, for each individual, 
\begin{equation}\label{foc}
\frac{\partial V_{i}\left(  x_{i}^{\ast},u_{-i}^{\ast}\right)  }{\partial
x_{i}}=0,
\end{equation}
and therefore, by $(\ref{focldg})$ with $j=i$, the induced utility
function $U_{i}$ is flat as a function of $x_{i}$ at $x^{\ast}$. Moreover, for
$j\neq i,$ it follows, again from $(\ref{focldg})$, that
\[
\begin{array}{c}
\frac{\partial U_{i}\left(  x^{\ast}\right)  }{\partial x_{i}\partial x_{j}} = \\ \\
\frac{\partial\left(  \frac{\partial U_{i}\left(  x^{\ast}\right)
}{\partial x_{i}}\right)  }{\partial x_{j}}=\frac{\partial\left(  \left(
\frac{\partial V_{i}\left(  x_{i}^{\ast},u_{-i}^{\ast}\right)  }{\partial
x_{i}}\right)  \left(  \left(  I-J_{x^{\ast}}\left(  u^{\ast}\right)  \right)
^{-1}\right)  _{ii}\right)  }{\partial x_{j}} = \\ \\
\frac{\partial V_{i}\left(  x_{i}^{\ast},u_{-i}^{\ast}\right)  }{\partial
x_{i}}\frac{\partial\left(  \left(  \left(  I-J_{x^{\ast}}\left(  u^{\ast
}\right)  \right)  ^{-1}\right)  _{ii}\right)  }{\partial x_{j}}%
+\frac{\partial V_{i}\left(  x_{i}^{\ast},u_{-i}^{\ast}\right)  }{\partial
x_{i}\partial x_{j}}\left(  \left(  I-J_{x^{\ast}}\left(  u^{\ast}\right)
\right)  ^{-1}\right)  _{ii}=0,
\end{array}
\]
where the last equality is due to the first order condition (\ref{foc}) coupled with the fact that
\[
\frac{\partial V_{i}\left(  x_{i}^{\ast},u_{-i}^{\ast}\right)  }{\partial
x_{i}\partial x_{j}}=0
\]
since $V_{i}$ does not depend on $x_{j}$. With\ a marginal change in $x_{j}$ from $x^{\ast}$, the function
$U_{i}$\ remains constant as a function of $x_{i}$, and $x_{i}^{\ast}$
therefore remains a local maximizer of $U_{i}$. 
\end{proof}

\subsection*{Stronger conditions}

Stronger conditions imply that $\left(  I-J_{x}\left(  u\right)  \right)  $ is a P-matrix.

\subsubsection*{Spectral radius less than one}

If the induced game $U$ is defined at $x,$ i.e. if $\left(  x,U\left(
x\right)  \right)  $ is consistent, then, by definition, $U\left(  x\right)
=V_{x}\left(  U\left(  x\right)  \right)  ,$ and therefore also $U\left(
x\right)  =V_{x}^{k}\left(  U\left(  x\right)  \right)  $ for every $k\geq
1$\footnote{\noindent$V_{x}^{k}$ is defined inductively by
\[
V_{x}^{1}\left(  u\right)  =V_{x}\left(  u\right)  ,\quad V_{x}^{k}\left(
u\right)  =V_{x}\left(  V_{x}^{k-1}\left(  u\right)  \right)
\]
for $k>1$.}. 
Moreover, given Assumption 2, the first equality holds
\textit{only} for $U\left(  x\right).  $By \citet*{galenikaido65}Theorem
4, if $u\neq U\left(  x\right)  $ then $V_{x}\left(  u\right)  \neq u$. \ 

Now, if $u$ is some small perturbation of $U\left(  x\right)  $, representing
a slight\ mis-assessment of the players regarding each other's utility levels
with the action profile $x,$ would the repeated re-assessments $V_{x}\left(
u\right)  ,$ $V_{x}\left(  V_{x}\left(  u\right)  \right)  ,...,V_{x}%
^{k}\left(  u\right)  ,...$ converge back towards $U\left(  x\right)  $?
\ This is a plausible requirement, because otherwise $U\left(  x\right)  $ is
an \textit{unstable} rest-point of $V_{x}$, and the definition of the induced
game $U$ is not robust to slight misperceptions.

The required convergence \
\[
V_{x}^{k}\left(  u\right)  \underset{k\rightarrow\infty}{\rightarrow}U\left(
x\right)
\]
is guaranteed in some small enough neighborhood of $U\left(  x\right).$ That is, 
$U\left(  x\right)  $ is an asymptotically stable fixed point of $V_{x}$ if
the spectral radius of $J_{x}\left(  U\left(  x\right)  \right)  $ \ (the
largest of the absolute values of its eigenvalues), denoted $\rho\left(
J_{x}\left(  U\left(  x\right)  \right)  \right)  ,$ satisfies\footnote{\citet*{galor07}, Theorem 4.8}
\[
\rho\left(  J_{x}\left(  U\left(  x\right)  \right)  \right)  <1,
\]
whereas if, in contrast, $\rho\left(
J_{x}\left(  U\left(  x\right)  \right)  \right)  >1$ and no eigenvalue of
$J_{x}\left(  U\left(  x\right)  \right)  $ has absolute value equal to 1,
then $V_{x}$ is not asymptotically stable, and diverges away from arbitrarily
small perturbations of $U\left(  x\right)  .$

\bigskip

In fact, the above re-assessments may take place among any subset
$I_{0}\subseteq I$ of the individuals, for fixed utility levels $\bar
{u}=\left(  \bar{u}_{j}\right)  _{j\in I\setminus I_{0}}$ of the remaining
individuals. The purely affective interaction $V$ defines a \emph{purely
affective sub-interaction} $V^{\bar{u}}$ among the individuals in $I_{0},$
\[
V^{\bar{u}}\left(  x,u\right)  =\left(  V_{i}\left(  x,u,\bar{u}\right),
\right)  _{i\in I_{0}},%
\]
where $x=\left(  x_{i}\right)  _{i\in I_{0}}$ \ and $u=\left(  u_{i}\right)
_{i\in I_{0}}$. The set of \emph{purely affective sub-interactions} of $V$ is
thus defined by ranging over all the non-empty subsets of individuals
$I_{0}\subseteq I$ \ and utility levels $\bar{u}=\left(  \bar{u}_{j}\right)
_{j\in I\setminus I_{0}}$ of the other individuals.

\bigskip

\begin{assumption}
For every $x\in X$\textit{ and }$u\in
R^{n},$%
\[
\rho\left(  J_{x}\left(  u\right)  \right)  <1,
\]
and the same holds for all the sub-interactions of $V$.
\end{assumption}

\bigskip

This assumption implies our Assumption 2.

\bigskip

\begin{proposition}
Under Assumption 4, for every $x\in X$ and $u\in\mathbb{R}^{n}$,  \newline $\left(
I-J_{x}\left(  u\right)  \right)  $ is a P-matrix.
\end{proposition}


\begin{proof} \ $\rho\left(  J_{x}\left(  u\right)  \right)  <1$
implies that all the eigenvalues $\lambda_{1},...,\lambda_{n}$ of
$J_{x}\left(  u\right)  $ are within the open unit disk around the origin of
the complex plane, and therefore that so are $-\lambda_{1},...,-\lambda_{n},$
which are the eigenvalues of $-J_{x}\left(  u\right)  $. Hence $1-\lambda
_{1},...,1-\lambda_{n},$ which are the eigenvalues of $I-J_{x}\left(
u\right)  ,$ all have positive real parts. These eigenvalues are the roots of
the characteristic polynomial of $I-J_{x}\left(  u\right)  $. This
characteristic polynomial has positive coefficients, and therefore its roots
are all either real, and therefore positive by the above, and/or come in
conjugate pairs of the form $c+di$ and $c-di$ whose product $c^{2}+d^{2}$ is
also positive. Hence the determinant of $I-J_{x}\left(  u\right)  ,$ which is
the product of its eigenvalues, is positive. \ 

All the above is true also for every sub-interaction involving only the subset
$I_{0}$ of individuals, implying the positivity of the determinant of the
principal submatrix of $I-J_{x}\left(  u\right)  $ with rows and columns in
$I_{0},$ i.e. the positivity of the principal minor with rows and columns in
$I_{0}$. \ We thus conclude that $I-J_{x}\left(  u\right)  $ is a P-matrix.
\end{proof}

\bigskip

\noindent\textbf{Remark. }The conclusion of proposition 1, i.e. Assumption 2,
is weaker than its premise, Assumption 4. For example, in the case of two
individuals, denoting
\[
J_{x}\left(  U\left(  x\right)  \right)  =\left(
\begin{array}
[c]{cc}%
0 & a\\
b & 0
\end{array}
\right)
\]
$\left(  I-J_{x}\left(  U\left(  x\right)  \right)  \right)  $ being a
P-matrix means $ab<1,$ whereas $\rho\left(  J_{x}\left(  U\left(  x\right)
\right)  \right)  <1$ means the more stringent requirement $\left\vert
ab\right\vert <1$. \ 

\bigskip

If $ab<-1$ then Assumption 2 holds, but Assumption 4 does not. \ In this case
the eigenvalues of $J_{x}\left(  U\left(  x\right)  \right)  $ are $\pm
\sqrt{ab}$, whose absolute values are both larger than 1, and therefore
$V_{x}$ diverges away from $U\left(  x\right)  $ from arbitrarily small
neighborhoods of $U\left(  x\right)  $.

\subsubsection*{Dominant Diagonal}

Another property of the matrix $(I - J_x(u)$ that we may consider is that the matrix is \emph{dominant diagonal}:

\begin{assumption}
For every $x$ and $u$, the matrix $(I - J_x(u))$ is dominant diagonal: there exists $ h(u) \in \mathbb{R}^n,$ such that, for any $i= 1, \ldots n,$
\[
h_i(u) > \sum_{j \neq i} h_j(u) \mid - \frac{\partial V_{x, i}}{\partial u_j} \mid. 
\]
\end{assumption}
That is, there is a way to rescale utilities at $u,$ such that marginal changes in $u_j$, for $j \neq i$,  have a total effect on $V_{x, i}$ less than 1.

\begin{proposition}
Under Assumption 5, for every $x\in X$ and $u\in\mathbb{R}^{n}$,
\newline
$ (I-J_{x}(u))  $ is a P-matrix\footnote{\citet*{moylan77}}.
\end{proposition} 

\section{Examples}

Examples illustrate the results and their implications.

\begin{example}[Non-separable affection]

Now let 
\begin{align*}
V_{1}\left(  x,u_{2}\right)   &  =x\left(  1-x\right)  -2xu_{2},\\
V_{2}\left(  y,u_{1}\right)   &  =y\left(  1-y\right)  +\frac{1}{8}yu_{1},
\end{align*}
where $x,y\in\left(  0,1\right)  $. In this example, individual 1 is rather
spiteful towards individual 2, and individual 2 is mildly sympathetic towards individual
1. For each $x,y\in\left(  0,1\right)  $ the Jacobian of $V$ is%
\[
J_{\left(  x,y\right)  }=\left(
\begin{array}
[c]{cc}%
0 & -2x\\
\frac{1}{8}y & 0
\end{array}
\right),
\]
whose eigenvalues are $\pm\frac{1}{2}i\sqrt{xy}$, and its spectral radius is
therefore $\frac{1}{2}\sqrt{xy}<1$. 

The induced game is
\begin{align*}
U_{1}\left(  x,y\right)   &  =\frac{8x\left(  (1-x)-2y\left(  1-y\right)
\right)  }{8+2xy},\\
U_{2}\left(  x,y\right)   &  =\frac{y\left(  8(1-y)+x\left(  1-x\right)
\right)  }{8+2xy},
\end{align*}
with the best reply functions
\begin{align*}
\beta_{1}\left(  y\right)   &  =\frac{2\sqrt{2y^{3}-2y^{2}+y+4}-4}{y},\\
\beta_{2}\left(  x\right)   &  =\frac{\sqrt{-2x^{3}+2x^{2}+16x+64}-8}{2x},
\end{align*}
whose intersection is the Nash equilibrium
\[
x=0.246\,20,\quad y=0.503\,79,
\]
where the reaction curves are locally flat.

%
\begin{center}

[Figure 1]

\includegraphics[scale=0.2]{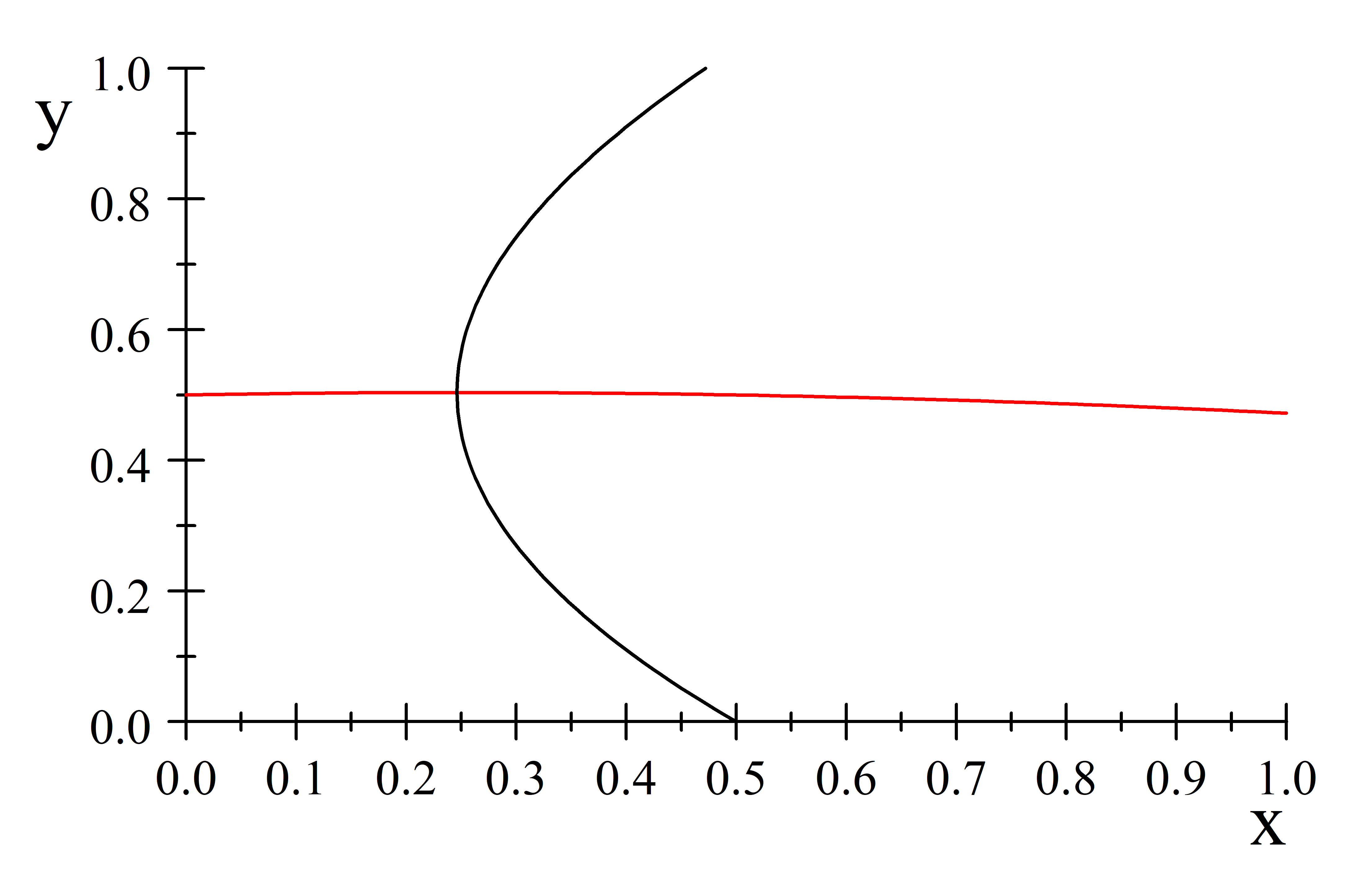}
\end{center}

For the spiteful individual 1, the Nash equilibrium is at a saddle point of his
utility function.
\begin{center}

[Figure 2]

\includegraphics[scale= 1]{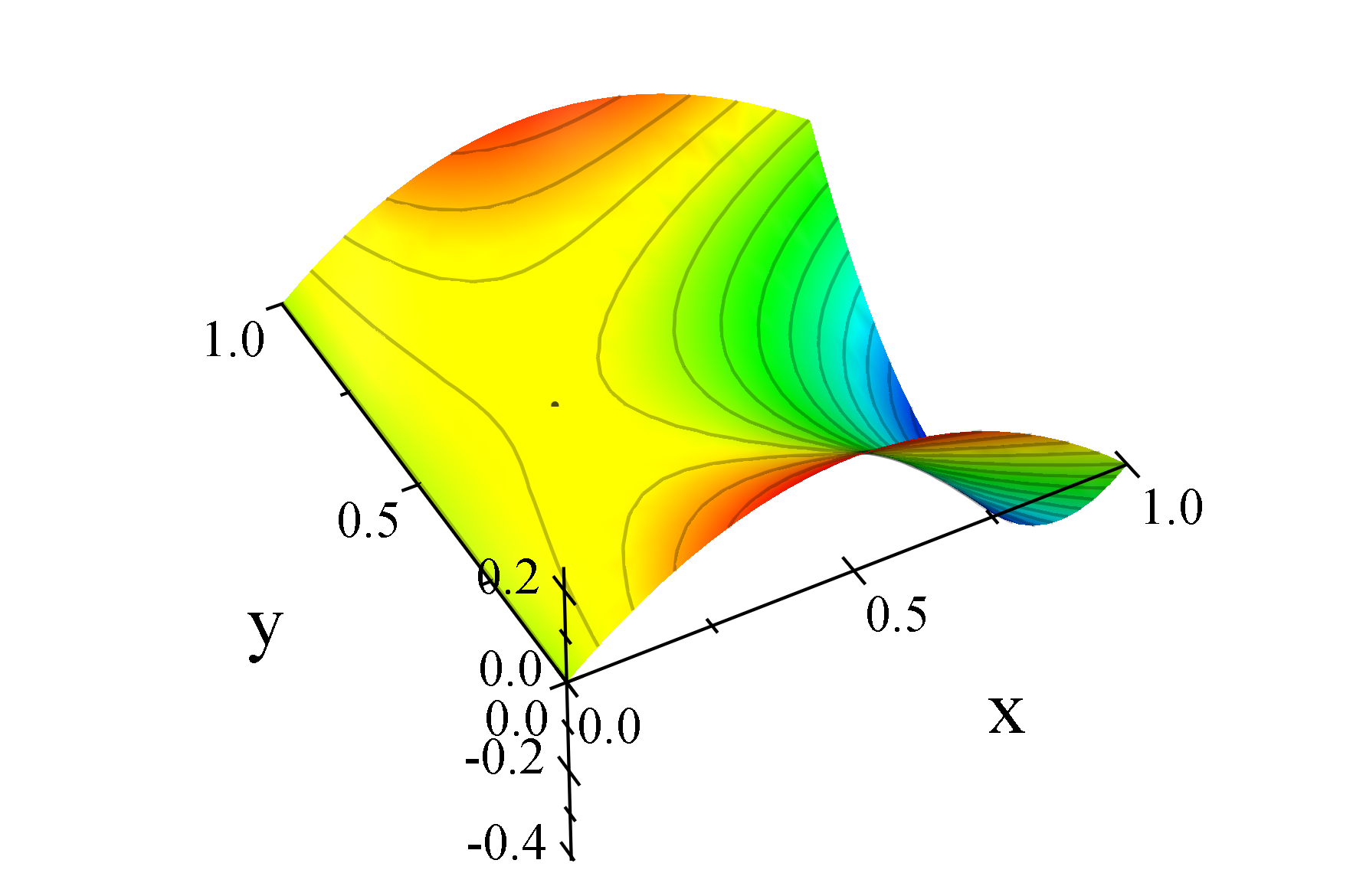}
\end{center}
For the sympathetic individual 2, the Nash equilibrium is at a hilltop\ of her
utility function.
\begin{center}

[Figure 3]

\includegraphics[scale=1]{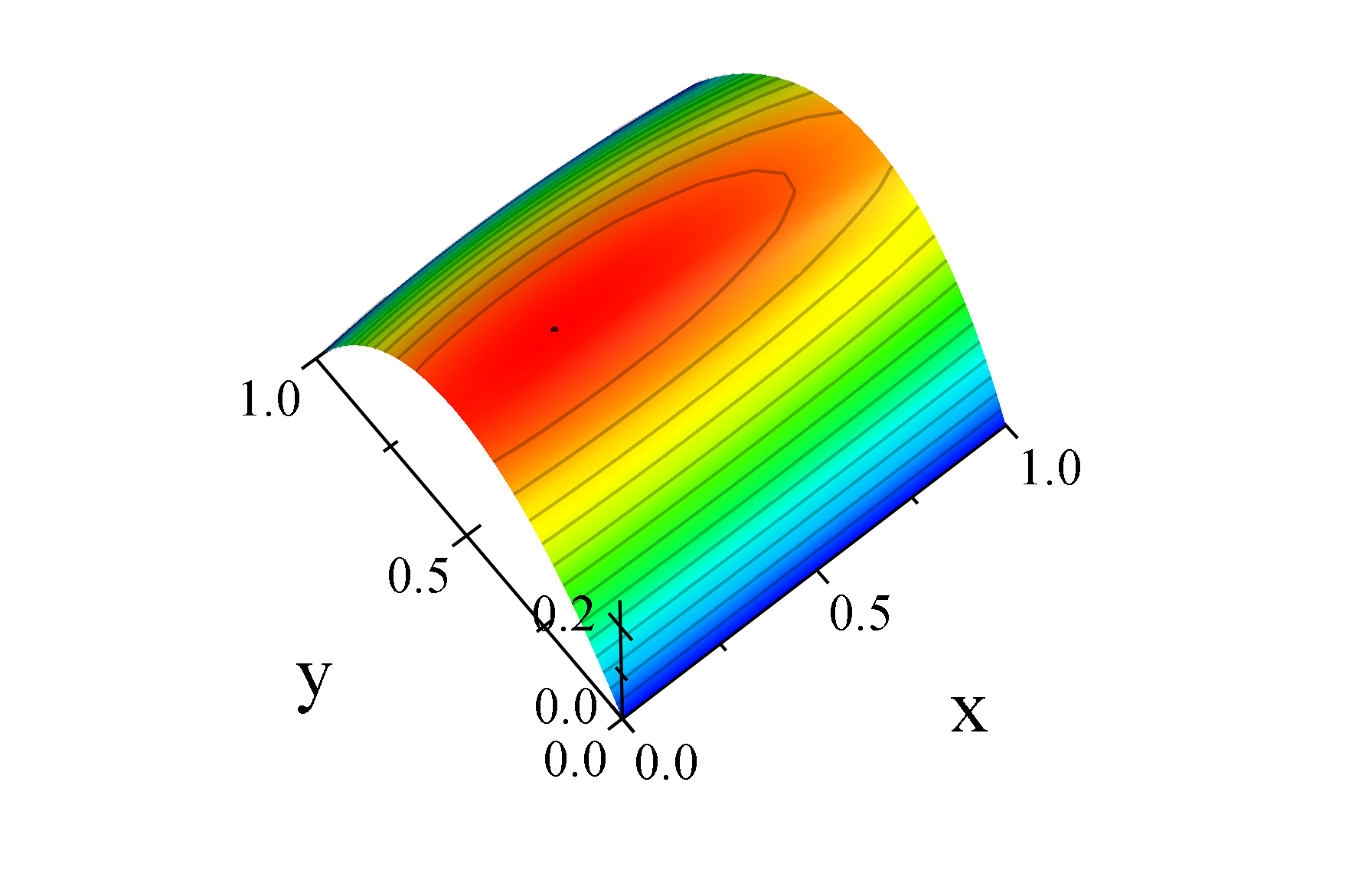}
\end{center}

\end{example}

\bigskip

\begin{example} [Shifting attitudes]

Next, 
\begin{align*}
V_{1}\left(  x,u_{2}\right)   &  =x^{2}\left(  1-x^{2}\right)  +\frac{1}%
{2}xu_{2},\\
V_{2}\left(  y,u_{1}\right)   &  =y^{2}\left(  1-y^{2}\right)  +\frac{1}%
{2}yu_{1},
\end{align*}
for $x,y\in\left(  -1,1\right)  ,$ so that each individual is sympathetic/spiteful
with positive/negative actions respectively. The Jacobian of $V$ at $\left(
x,y\right)  $ is
\[
J_{\left(  x,y\right)  }=\left(
\begin{array}
[c]{cc}%
0 & \frac{1}{2}x\\
\frac{1}{2}y & 0
\end{array}
\right),
\]
whose eigenvalues are\ $\pm\frac{1}{2}\sqrt{xy}.$ This implies that  the spectral
radius of the Jacobian is smaller than $\frac{1}{2}$. 

The induced game is
\begin{align*}
U_{1}\left(  x,y\right)   &  =\frac{4x^{2}\left(  1-x^{2}\right)
+2xy^{2}\left(  1-y^{2}\right)  }{4-xy},\\
U_{2}\left(  x,y\right)   &  =\frac{4y^{2}\left(  1-y^{2}\right)
+2yx^{2}\left(  1-x^{2}\right)  }{4-xy}.%
\end{align*}

The graphs of $U_{1}$  and $U_{2}$  are in 
Figure 4 and Figure 5.
%
%
\begin{center}

[Figure 4 ]

\includegraphics[scale=1]{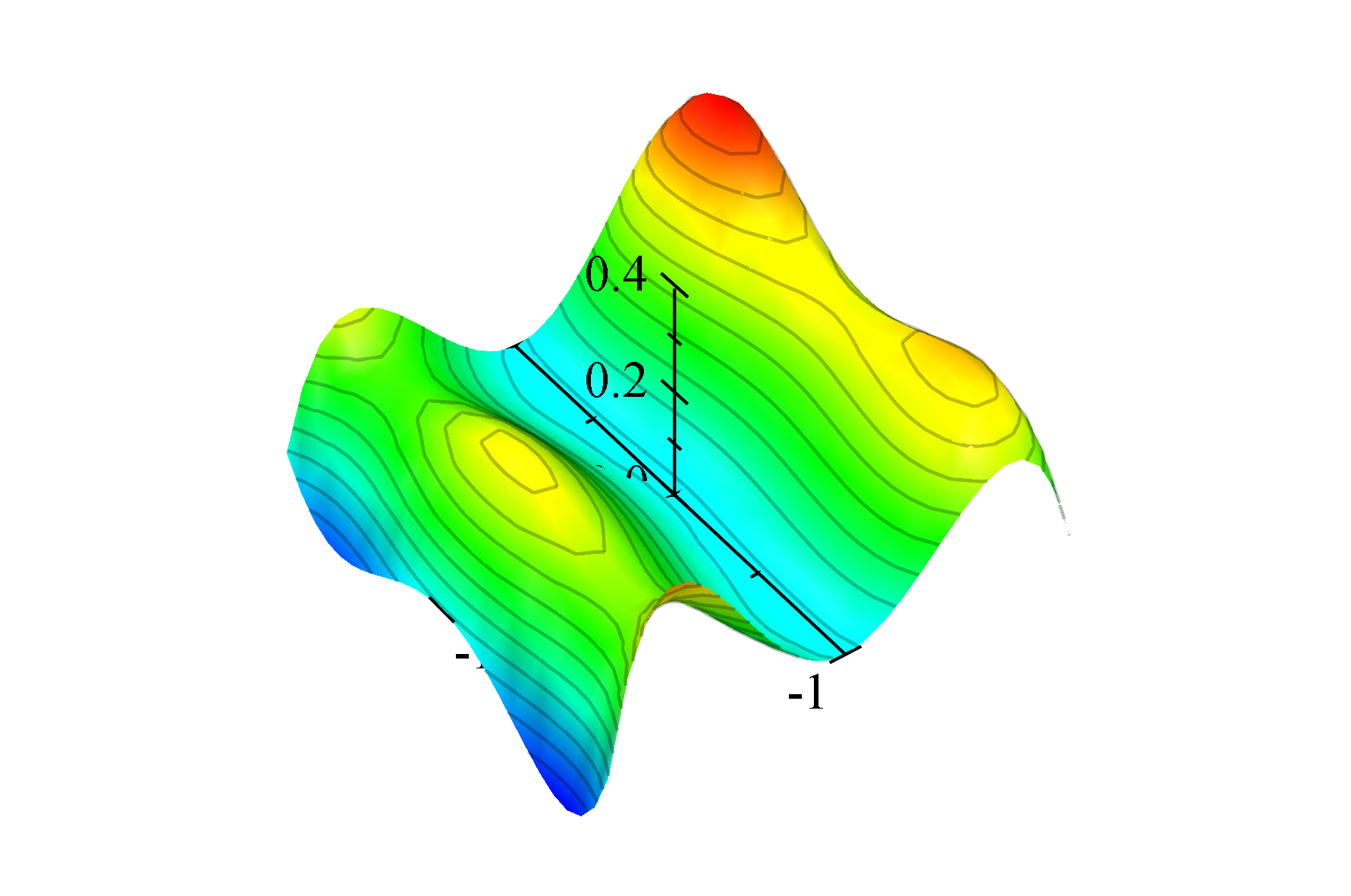}
\end{center}

\begin{center}

[Figure 5]

\includegraphics[scale=1]{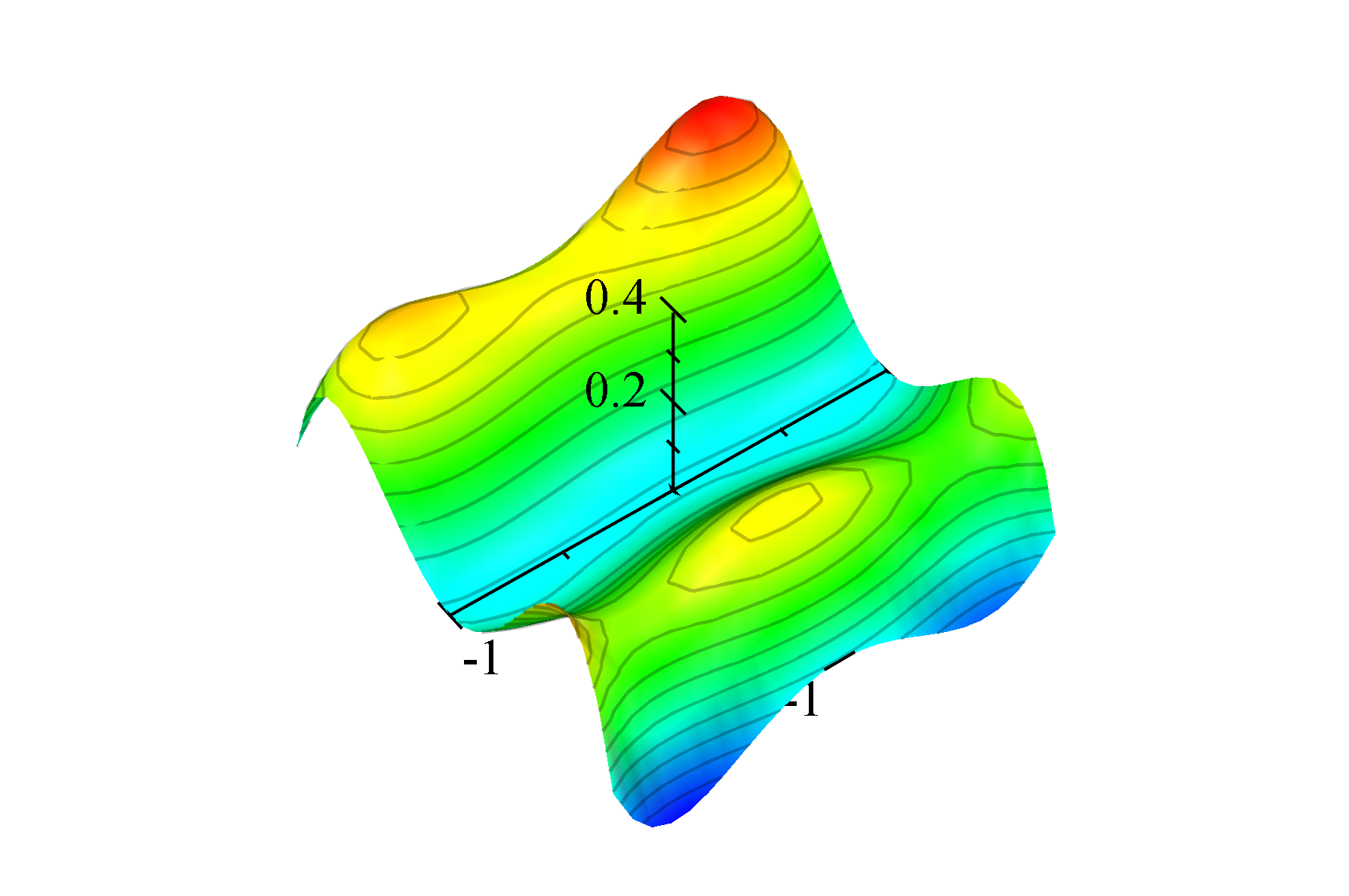}
\end{center}

\quad The unique Nash equilibrium is%
\[
x=y=0.751\,97.
\]
Both individuals are sympathetic and the Nash equilibrium is at the peak of their
utility functions. The Nash equilibrium is Pareto optimal, and it maximizes the
average of their utilities.

\bigskip

If, instead, the individuals were confined to negative actions 
\newline
$x,y\in\left(
-1,0\right)  ,$ the unique Nash equilibrium would be%
\[
x=y=-0.68266
\]
that is at a saddle point of $U_{1}$ and of $U_{2},$ and maximizes the
individuals' average utility in that quadrant but not globally. 

\bigskip
Similarly, if
individual $1$ were confined to positive actions (and thus sympathy) while individual
2 to negative actions (spite), there would be a unique Nash equilibrium within
that quadrant%
\[
x=0.724\,71,\quad y=-0.665\,76
\]
with individual $1$ at a hilltop and individual $2$ at a saddle point, maximizing the
average utility within that quadrant, but not globally.

\end{example}

\section{Economies with affective interaction}

\citet*{pearce08} showed that, in cake-eating games with positive interdependent
affect, the subgame perfect path is \textit{not} Pareto optimal. This is not
in contrast with Corollary 1, though, because in cake-eating games, the
available choices of subsequent individuals are restricted by those of the
preceding ones. Put differently, if an individual was to be assigned a very negative
utility if she tried to consume more than the remainder of the cake, the utility functions $V_{i}$ would depend on the actions of other players (and cease to be smooth). In yet other words, cake-eating games are not genuine games, but rather economies
with sequential consumption.

\bigskip

To understand the nature of the problem induced by the presence of a feasibility constraint, we consider a slightly modified version of \citet*{bergstrom89} `Love and Spaghetti' example.

Romeo and Juliet both  care about good $x$, spaghetti, and the other's happiness, and 
\[
\begin{array}{ccc}
V_{1} (x_1, u_2)& = &  \sqrt{x_1} + a u_2 ,  \\ & & \\ 
V_{2} (x_2, u_1) & = & \sqrt{x_2}  + b u_1.
\end{array}
\]

Differently from Bergstrom, we allow for negative $a$ and $b$. As in the general model, we assume that  $ab < 1$ (so here Romeo and Julliet  do not necessarily love each other, but their affective interdependence, positive or negative, is moderate).
\bigskip

We also add a second good, money, entering quasi-linearly in the utility function, and assume each member of the couple has an initial endowment $e_i = (1, M)$.

We can solve for the induced utilities $U_1 (x, m)$ and $U_2 (x, m), $ and 
\[
\begin{array}{ccc}
U_{1}(x, m)  & = &  \frac{1}{1 - ab} \: ( \sqrt{x_1} + m_1) + \frac{a}{1 - ab} \:  (\sqrt{x_2} + m_2) , \\ & & \\
U_{2}(x, m)& = & \frac{b}{1 - ab} \:  ( \sqrt{x_1} + m_1)  + \frac{1}{1 - ab} \:  (\sqrt{x_2} + m_2).
\end{array}
\]

At a \emph{competitive equilibrium of an economy with affective interaction}, each individual chooses $(x_i, m_i)$ to maximize $U_i$ taking $(x_j, m_j)$  as given, under the budget constraint: $ p_x x_i + p_m m_i = p_x + p_m M$, and prices adjust to guarantee feasibility.

If we fix $p_m=1$ (and $M$ is large enough), at the unique competitive equilibrium,
$$ \hat p_x=1, $$ 
$$\hat x_i = e_i =1, \quad i =1, 2.$$

A {\emph{benevolent non-myopic social planner}} chooses  the allocation $(x_1, x_2)$ of spaghetti to solve
\[\max_{x_1, x_2} W_{\lambda} (x_1, x_2, m) = \sum_i \lambda_i U_i(x_1, x_2, m),
\]
under  the constraint
$$ 
x_1 + x_2 = 2,
$$
leading to the first order condition
\begin{equation}\label{plannerFOC}
\frac{\sqrt{x_1}}{\sqrt{2 - x_1}}= \frac{\lambda _1 + \lambda_2 b}{\lambda_1 a + \lambda_2}.
\end{equation}

We see that, in  an economy,  moderate reciprocal affection, $ab < 1$ \emph{does not} guarantee that there exist $(\lambda _1, \lambda_2) > 0$ such that the equilibrium allocation $\hat x_1 = 1$ solves the planner problem. 

For example,  if $a=2$ and $b=1/4,$ at $\hat x_1 = 1$, (\ref{plannerFOC}) becomes $\lambda_1 = - \frac{3}{4} \lambda_2$.
Indeed, the equilibrium allocation of good, $\hat x_1 = \hat x_2=1$,  generates utilities $\hat u_1= 2\sqrt{\hat x_1} + 4 \sqrt{\hat x_2} = 6$, $\hat u_2= \frac{1}{2}\sqrt{\hat x_1} + 2 \sqrt{\hat x_2} = 2.5$, while a planner solving (\ref{plannerFOC}) with $\lambda_1 = \lambda_2 =1$ would rather choose $\tilde x_1 = 0.29$, $\tilde x_2= 2 - \tilde x_1 = 1.71$, generating utilities $\tilde u_1= 2\sqrt{\tilde x_1} + 4 \sqrt{\tilde x_2} = 6.28$ and $\tilde u_2= \frac{1}{2}\sqrt{\tilde x_1} + 2 \sqrt{\tilde x_2} = 2.87$, a Pareto improvement: the planner reallocates the good taking into account the strong love of individual $1$ towards individual $2$. 

For the specific preferences and endowments of the example,  positive welfare weights such that the competitive equilibrium maximizes social welfare, and is therefore Pareto optimal,  do exist under the stronger condition: $ab<1$  \emph{and}  $a<1$, $b<1$.  A characterization of economies with  affective interaction in which competitive equilibrium allocations are Pareto optimal remains an open problem.

\bibliography{ave-bibliography}

   	\end{document}